\documentclass[11pt]{article}
\usepackage{amsmath}
\usepackage{amsthm}
\usepackage{amsfonts}
\usepackage{enumitem}  
\usepackage{amssymb}
\usepackage{array}
\usepackage[utf8]{inputenc}
\usepackage{graphicx}
\usepackage{algorithm}
\usepackage{algorithmic}
\usepackage{authblk}
\usepackage{geometry}
\usepackage[english]{babel}
\usepackage{xcolor}

\newtheorem{theorem}{Theorem}[section]

\newtheorem{corollary}[theorem]{Corollary}

\theoremstyle{definition}
\newtheorem{definition}[theorem]{Definition}
\newtheorem{example}[theorem]{Example}
\newtheorem{remark}[theorem]{Remark}

\newcommand{\F}{\mathbb{F}}

\newcommand{\C}{\mathcal{C}}
\newcommand{\N}{\mathbb{N}}

\newcommand{\wt}{\mathrm{wt}}

\newcommand{\floor}[1]{{\left\lfloor{#1}\right\rfloor}}

\newcommand{\dfree}{\mathrm{d_{free}}}

\newcommand{\mylabel}[2]{#2\def\@currentlabel{#2}\label{#1}}

\title{On the left primeness of some polynomial matrices with applications to convolutional codes}

\author{Gianira N. Alfarano}
\author{Julia Lieb}
\affil{University of Z{u}rich, Switzerland}

\begin{document}

\maketitle

\begin{abstract}
    Maximum distance profile (MDP) convolutional codes have the property that their column distances are as large as possible for given rate and degree. There exists a well-known criterion to check whether a code is MDP using the generator or the parity-check matrix of the code.
    In this paper, we show that under the assumption that $n-k$ divides $\delta$ or $k$ divides $\delta$, a polynomial matrix that fulfills the MDP criterion is actually always left prime. In particular, when $k$ divides $\delta$, this implies that each MDP convolutional code is noncatastrophic. Moreover, when $n-k$ and $k$ do not divide $\delta$, we show that the MDP criterion is in general not enough to ensure left primeness. In this case, with one more assumption, we still can guarantee the result.
   
 \end{abstract}

\section{Introduction}
In the algebraic theory of error correcting codes, one important family of codes for telecommunication is given by convolutional codes. A convolutional code is a submodule of rank $k$ of $\F[z]^n$, where $\F$ is a finite field. Unfortunately, constructions of these codes, having some good designed minimum distance, are quite rare. Convolutional codes have the flexibility of grouping blocks of information in an appropriate way, according to the erasures location, and then decoding first the part of the sequence with least erasures or the part of the sequence where the distribution of erasures allows a complete correction.

The classical notion of minimum distance for a convolutional code would be the so-called \emph{free distance}. However, for convolutional codes, there is a notion of distance that is more important than the free distance, namely the \emph{$j$-th column distances}, which satisfy a set of upper bounds given in \cite{gl03}. If as many of these column distances as possible meet these bounds with equality, then $\C$ is called \emph{maximum distance profile (MDP)}, see \cite{gl03}. It was shown in \cite{virtu2012} that MDP convolutional codes have optimal recovery rate for windows of a certain length, depending on the code parameters.

As MDP convolutional codes have the maximal possible growth in the column distances, they can correct the maximal number  of errors in a time interval, and therefore are similar to MDS block codes within windows of fixed size.
However, in contrast to the case of MDS block codes, there are very few algebraic constructions of MDP convolutional codes, all based on a characterization provided in \cite{gl03}, which works with the assumption that a parity-check matrix $H(z)$ (or equivalently a generator matrix $G(z)$) of the convolutional code is left prime. In particular, this criterion says that a left prime matrix $H(z)\in\F[z]^{(n-k)\times n}$ is a parity-check matrix  (equivalently a left prime $G(z)\in\F[z]^{k\times n}$ is a generator matrix) of an MDP convolutional code if all the full-size minors of the truncated sliding parity-check matrix $H_j^c$ (equivalently the sliding generator matrix $G_j^c$) that are not trivially zero are nonzero.

Left prime matrices have been widely investigated in module theory, in system theory and also in the theory of convolutional codes. In the literature one can find several properties equivalent to left primeness (see, for instance, Chapter 6 of \cite{kailath1980linear}). However, proving even one of them for a given polynomial matrix is in general not easy.

Several papers that provide a (concrete) construction for MDP convolutional codes, for example \cite{virtu2012}, \cite{AlmeidaNappPinto2013}, \cite{li17}, are based on the characterization for the parity-check matrix mentioned above. Unfortunately, in all of them there is no discussion on the left primeness of the constructed matrices. Indeed, in all the mentioned works, only the criterion on the minors of the sliding parity-check is shown to be satisfied.
In this paper, 
we explain in a remark why the left primeness is not needed in order that this criterion is valid and thus, all of these constructions are correct. However, in general it is not easy to compute the degree of a convolutional code from a parity-check matrix that is not left prime and hence, it is not a priori clear that the constructed codes have really the degree that is stated in these papers.

In this paper, we show that if $H(z)$ is a parity-check matrix (resp. $G(z)$ is a generator matrix) of an $(n,k,\delta)$ convolutional code $\C$, where $n-k$ divides $\delta$ or $k$ divides $\delta$ and such that the criterion on the minors of the truncated sliding parity-check matrix $H_L^c$ (resp. generator matrix $G_L^c$) of $\C$ is satisfied, then $H(z)$ (resp. $G(z)$) is left prime. Observe that if $n-k$ divides $\delta$, we consider the parity-check matrix and if $k$ divides $\delta$, we consider the generator matrix. If $k$ divides $\delta$, our result implies that all $(n,k,\delta)$ MDP convolutional codes are necessarily noncatastrophic. If $n-k$ divides $\delta$, it implies that a polynomial matrix $H(z)$ that fulfills the criterion is a parity-check matrix of a convolutional code whose degree equals the sum of the row degrees of $H(z)$ (and of course is noncatastraphic as it has a parity-check matrix). 

If $n-k$ and $k$ do not divide $\delta$, then we require some technical assumption in addition to the criterion on the minors of the sliding matrices to obtain that the left primeness property is satisfied.

The paper is structured as follows. In Section \ref{sec:preliminaries}, we give some background on convolutional codes, focusing on the family of MDP convolutional codes. Moreover, we point out the importance of the left primeness property for generator or parity-check matrices of noncatastrophic convolutional codes. In Section \ref{sec:mainresult}, we prove the main result of the paper, namely we show that the MDP property on the truncated sliding parity-check matrix $H_L^c$ of a convolutional code, implies that $H(z)$ is left prime. The same result can be shown for the generator matrix of a convolutional code. We conclude with some more remarks in Section \ref{sec:conclusion}.


\section{Preliminaries}\label{sec:preliminaries}
In this section we give the basic notions and results on the theory of convolutional codes. For a more detailed treatment we refer to \cite{lieb2020convolutional}. In particular, we will focus on the property of left primeness of a parity-check (or generator) matrix of a noncatastrophic convolutional code.

Let $\F$ be a finite field, $\F[z]$ be the polynomial ring over $\F$ in the indeterminate $z$ and let $k,n$ be positive integers, such that $k< n$. An $(n,k)$ \emph{convolutional code} $\C$ is a submodule of $\F[z]^n$ of rank $k$. 

Since $\F[z]$ is a principal ideal domain, every submodule of $\F[z]^n$ is free. Hence, there exists a polynomial matrix $G(z)\in \F[z]^{k\times n}$, whose rows are a basis for $\C$. We call such $G(z)$ a \emph{generator matrix} of the convolutional code $\C$. Note that it is not unique. Assume that $G(z)\in\F[z]^{k\times n}$ and $\bar{G}(z)\in\F[z]^{k\times n}$ are two generator matrices for the same convolutional code $\C$, then, there exists a \emph{unimodular} matrix $U(z)\in\F[z]^{k\times k} $ such that $\bar{G}(z)=U(z)G(z)$ and, we say that $G(z)$ and $\bar{G}(z)$ are \emph{equivalent}.

Given an $(n,k)$ convolutional code $\C\subseteq\F[z]^n$, with generator matrix $G(z)\in\F[z]^{k\times n}$, we define the largest degree among the entries in the $i$-th row of $G(z)$ as the \emph{$i$-th row degree} and we denote it by $\delta_i$. We say that $G(z)$ is \emph{row-reduced} or \emph{minimal} if the sum of its row degrees attains the minimum possible value. We define the \emph{degree} $\delta$ of $\C$ to be the highest degree of the $k\times k$ minors of $G(z)$. When the degree $\delta$ of an $(n,k)$ convolutional code $\C\subseteq \F[z]^n$ is known, we denote $\C$ by $(n,k,\delta)$. Note that if $G(z)$ is row-reduced, then $\delta=\delta_1 + \dots +\delta_k$. Moreover, given a generator matrix $G(z)$ of a convolutional code $\C$, there always exists a row-reduced generator matrix equivalent to $G(z)$. 

There is another important property of matrices that is useful in the context of convolutional codes.

\begin{definition}\label{def:leftprime}
Let $G(z)\in \F[z]^{k\times n}$ be a matrix. Then $G(z)$ is said to be \emph{left prime} if in all the factorizations $G(z) = M(z)\bar{G}(z)$, with $M(z)\in\F[z]^{k\times k}$ and $\bar{G}(z)\in\F[z]^{k\times n}$, the left factor $M(z)$ is unimodular, i.e. $M(z)\in \mathrm{GL}_k(\mathbb F[z])$.
\end{definition}

There are several characterizations for left prime matrices. In particular, it holds that $G(z)\in\F[z]^{k\times n}$ is left prime if and only if it admits a right $n\times k$ polynomial inverse (see Chapter 6 of \cite{kailath1980linear} for details).

Since all generator matrices of a convolutional code $\C$ are equivalent up to multiplication by unimodular matrices, if $\C$ admits a left prime generator matrix, then all its generator matrices are left prime. In this case, we  say that $\C$ is a \emph{noncatastrophic} convolutional code.

Let $\C$ be a noncatastrophic $(n,k,\delta)$ convolutional code and  $G(z)\in \F[z]^{k\times n}$ be a generator matrix of $\C$. Then there exists a matrix $H(z)\in\F[z]^{(n-k)\times n}$, such that
\begin{equation}
c(z)\in\C \Leftrightarrow H(z)c(z)^\top=0.
\end{equation} 
Such a matrix $H(z)$ is called a \emph{parity-check matrix} of $\C$. In \cite{york1997algebraic}, it has been shown that a convolutional code $\C$ is noncatastrophic if and only if it admits a parity-check matrix. 

Also a parity-check matrix of a convolutional code is not unique. In fact, every convolutional code has several left prime parity-check matrices  and several parity-check matrices that are not left prime (in contrast to generator matrices where left primeness is a property of a noncatastrophic code). Indeed, if we consider a left prime parity-check matrix for a convolutional code and multiply it from the left with any polynomial matrix, we obtain another parity-check matrix for the same code.
Moreover, in \cite{ro01}, it is shown that if $H(z)\in\F[z]^{(n-k)\times n}$ is a left prime and row-reduced parity-check matrix of an $(n,k,\delta)$ convolutional code $\C$, then the sum of the row degrees of $H(z)$ is equal to $\delta$. This is not true in general. Indeed,
 the following example shows that if  a not left prime parity-check matrix $H(z)$ of a convolutional code $\C$ is given, one can not obtain the degree of $\C$ as sum of the row degrees of $H(z)$.

\begin{example}\label{ex:degreeleftprime}
Let $\C$ be a $(3,1)$ convolutional code with with parity-check matrix $$H(z)=\left[\begin{array}{ccc}z(1+z) & 0 & 1+z\\
    0 & 1+z & 1+z\end{array}\right].$$ 
    Observe that $\C$ has degree $1$ since the matrix $$\tilde{H}(z)=\left[\begin{array}{ccc}z & 0 & 1\\
    0 & 1 & 1\end{array}\right]$$ is a left prime and row-reduced parity-check matrix of the same convolutional code, but the sum of the row degrees of $H(z)$ is $3$. Moreover, the maximal degree of the full-size minors of $H(z)$ is also $3$. This shows that the only way to obtain the degree of the code is by computing an equivalent left prime parity-check matrix. Note also that it does not help that $H(z)$ is row-reduced and that $H(0)$ has full rank.
\end{example}

Let $\C\subseteq\F[z]^n$ be an $(n,k,\delta)$ convolutional code. Thanks to the isomorphism between $\F[z]^n$ and $\F^n[z]$, we can define a weight function on $\C$ as follows. Given a codeword $v(z)=\sum_{i=0}^r v_iz^i\in \C$, we define the \emph{weight} of $v(z)$ as
$$\wt(v(z)) := \sum _{i=0}^r \wt(v_i)\in\N_0,$$
where $\wt(v_i)$ denotes the Hamming weight of $v_i\in\F^n$, i.e. the number of its nonzero components.
Finally, the \emph{free distance} of a convolutional code $\C$ is defined as
$$\dfree(\C):=\min\{\wt(v(z))\mid v(z)\in\C, v(z)\ne 0\}.$$
The generalized Singleton bound for an $(n,k,\delta)$ convolutional code $\C$, derived by Rosenthal and Smarandache in \cite{ro99a1}, relates the parameters of a convolutional code via the following inequality:

\begin{equation}\label{eq:genSingleton}
\dfree(\C)\leq (n-k)\left(\floor{\frac{\delta}{k}}+1\right)+\delta +1.
\end{equation}

A convolutional code whose free distance reaches the bound \eqref{eq:genSingleton} with equality  is called \emph{maximum distance separable (MDS) convolutional code}.


\subsection{MDP convolutional codes}
In this section we briefly define what MDP convolutional codes are and why their study is important. 

 In the context of convolutional codes, one aims to build codes which can correct as many errors as possible within windows of different sizes. This property is described by the notion of column distances. More formally, we introduce the following notation. Let $v(z) = \sum_{i=0}^r v_iz^i\in\F^n[z]$. For any positive integer $j\leq r$, let $v_{[0,j]}(z) := \sum_{i=0}^j v_iz^i$.

\begin{definition}

The \emph{j-th column distance} $d_j^c$ of an $(n,k,\delta)$ convolutional code $\C$ is defined as
$$ d_j^c :=\min\{\wt(v_{[0,j]}(z))\mid v(z)\in\C, \quad v(z)\ne 0\} $$
\end{definition}

Moreover, the column distances of  $\C$ satisfy the following set of bounds.

\begin{theorem}\textnormal{\cite[Proposition 2.2]{gl03}\label{thm:bound_dj}}
For every integer $j\in\N_0$,
\begin{equation}\label{eq:col_dist_bound}
d_j^c\leq(n-k)(j+1)+1.
\end{equation}
\end{theorem}

\begin{corollary}\textnormal{\cite[Corollary 2.3]{gl03}}\label{cor:boundfori<j} 
If $d_j^c\leq(n-k)(j+1)+1$ for some $j\in\N_0$, then $d_i^c\leq(n-k)(i+1)+1$ for every $i<j$.
\end{corollary}

Obviously, $d_j^c\leq \dfree(\C)$ for every $j$. It is easy to see that the maximum index for which the bound (\ref{eq:col_dist_bound}) is achievable is for $j=L$, where $$L:=\floor{\frac{\delta}{k}} + \floor{\frac{\delta}{n-k}}.$$ The $(L+1)$-tuple of numbers $(d_0^c,\dots, d_L^c)$ is called the \emph{column distance profile} of the code $\C$.

\begin{definition}
An $(n,k,\delta)$ convolutional code $\C$  whose column distances $d_j^c$ meet the bound of Theorem \ref{thm:bound_dj} with equality, for all $j=0,\dots,L$,  is called \emph{maximum distance profile} (MDP).
\end{definition}

Recall that the encoding map of an $(n,k,\delta)$ convolutional code $\C$ is given by the action of a polynomial matrix $G(z)$ and it can be expressed via the multiplication by the following polynomial:
$$G(z):=G_0+G_1z +\dots + G_mz^m,$$
where $G_i\in\F^{k\times n}$ and $G_m\ne 0$. 
In the same way, the parity-check matrix is given by
$$H(z):=H_0+H_1z +\dots + H_\mu z^\nu, ,$$
with $H_i\in\F^{(n-k)\times n}$ and $H_\nu\ne 0$.

Let $\C$ be an $(n,k,\delta)$ convolutional code, $G(z)$ be a generator matrix of $\C$ and  $H(z)$ be a parity-check matrix for $\C$. For any $j\in\N_0$, we define the \emph{$j$-th truncated sliding generator matrix}
and the \emph{$j$-th truncated sliding parity-check matrix} as
\begin{align*}
& G_j^c :=\begin{pmatrix}
G_0 & G_1 & \cdots  & G_j \\
 & G_0 & \cdots & G_{j-1} \\
 & & \ddots & \vdots \\
 & &  & G_0\\
\end{pmatrix}\in\F^{(j+1)k\times (j+1)n},\\
& H_j^c :=\begin{pmatrix}
H_0 & & &  \\
H_1 & H_0 & \\
\vdots & \vdots & \ddots & \\
H_{j} & H_{j-1} & \cdots & H_0
\end{pmatrix}\in\F^{(j+1)(n-k)\times (j+1)n},
\end{align*}
where $G_j = 0,$ whenever $j>m$ and $H_j=0$ whenever $j>\nu$.

These sliding matrices are relevant for the following characterization of MDP convolutional codes.

\begin{theorem}\textnormal{\cite[Corollary 2.3 and Theorem 2.4]{gl03}}\label{thm:characterizationMDP}
Let $G(z) = \sum_{i=0}^m G_iz^i$ and $H(z) = \sum_{i=0}^nu H_iz^i$  be a left prime generator matrix and a left prime  parity-check matrix, respectively, of an $(n,k,\delta)$ convolutional code $\C$. The following are equivalent:
\begin{enumerate}
  \item  $d_j^c(\C)=(n-k)(j+1)+1$\label{itm:cond1},
  \item every $(j+1)k \times (j+1)k$ full-size minor of $G_j^c$ formed by the columns with indices $1\leq t_1< \dots < t_{(j+1)k}$, where $t_{sk+1}>sn$ for $s=1,\dots,j$, is nonzero \label{itm:cond2},
  \item  every $(j+1)(n-k) \times (j+1)(n-k)$ full-size minor of $H_j^c$ formed by the columns with indices $1\leq t_1< \dots < t_{(j+1)(n-k)}$, where $t_{s(n-k)+1}\leq sn$ for $s=1,\dots,j$, is nonzero.\label{itm:cond3}
\end{enumerate}
In particular, $\C$ is MDP if and only if one of the above equivalent conditions holds for $j= L$.
\end{theorem}

Observe that the minors considered in Theorem \ref{thm:characterizationMDP} are the only full-size minors of $G_j^c$ and $H_j^c$ that can possibly be non-zero. For this reason, we call these minors \emph{non trivially zero}. 

In the following, we refer to the conditions \ref{itm:cond2} and \ref{itm:cond3} of Theorem \ref{thm:characterizationMDP} as MDP property \ref{itm:cond2} (on the sliding generator) and MDP property \ref{itm:cond3} (on the parity-check matrix) of a convolutional code. 

\begin{remark}

In Theorem \ref{thm:characterizationMDP} we assume that $G(z)$ and $H(z)$ are left prime. We will explain the exact role of this property:
\begin{enumerate}[label=(\roman*)]
    \item Considering the corresponding proof in \cite{gl03}, one observes that for the equivalence between conditions \ref{itm:cond1}  and \ref{itm:cond2} it is in fact enough to assume that $G_0$ is full rank (which is a consequence of $G(z)$ being left prime). However, both \ref{itm:cond1} and \ref{itm:cond2} imply that $G_0$ is full rank. For \ref{itm:cond1} this is true, because for $j=0$ this means that $G_0$ is the generator matrix of an MDS block code, i.e. in particular full rank. For \ref{itm:cond2} this follows immediately from the structure of $G_j^c$.
    Hence, it is possible to get rid of the assumption that $G(z)$ is left prime. However, note that if $G(z)$ is not left prime, the corresponding code is catastrophic. 
    
    \item Now we consider the equivalence between \ref{itm:cond1} and \ref{itm:cond3}, which of course is only possible if the code has a parity-check matrix, i.e. is noncatastrophic.
    If $H(z)$ is not left prime, then there exists an equivalent row-reduced and left prime parity-check matrix for the code $\tilde{H}(z)$, such that $H(z)=U(z)\tilde{H}(z)$ with $U(z)\in\mathbb F[z]^{(n-k)\times(n-k)}$ and $\deg(\det U(z)))>0$. Hence, with $U(z)=\sum_{i}U_iz^i\in\F^{(n-k)\times (n-k)}[z]$ and $U_i=0$ for $i>\deg(U(z))$, one has
    $$\left[\begin{array}{ccc}H_0 & & 0\\ \vdots & \ddots & \\ H_j & \cdots & H_0\end{array}\right]=\left[\begin{array}{ccc}U_0 & & 0\\ \vdots & \ddots & \\ U_j & \cdots & U_0\end{array}\right]\left[\begin{array}{ccc}\tilde{H}_0 & & 0\\ \vdots & \ddots & \\ \tilde{H}_j & \cdots & \tilde{H}_0\end{array}\right],$$
    for all $j\in\mathbb N_0$. Since $\tilde{H}(z)$ is left prime, $\tilde{H}_0$ is full rank. If $H(z)$ fulfills \ref{itm:cond3}, then all the full-size minors of $H_0$ are nonzero. Together with $H_0=U_0\tilde{H}_0$, this implies that $U_0$ and $\left[\begin{array}{ccc}U_0 & & 0\\ \vdots & \ddots & \\ U_j & \cdots & U_0\end{array}\right]$ are full rank. Consequently, $H_j^c$ fulfills \ref{itm:cond3} if and only if $\tilde{H}_j^c$ fulfills \ref{itm:cond3}, and since $\tilde{H}(z)$ is left prime, $H(z)$ and $\tilde{H}(z)$ are parity-check matrices of an MDP convolutional code whose degree $\delta$ is equal to the sum of the row degrees of $\tilde{H}(z)$. Hence, also for the implication from \ref{itm:cond3} to \ref{itm:cond1}, it is not necessary that the parity-check matrix of the code is left prime. 
    
    However, to construct an MDP convolutional code with a given $\delta$ it is necessary to construct it via a left prime parity-check matrix. Otherwise we do not know the degree of the constructed code since it is in general not an easy task to determine the degree of a convolutional code if we only know one of its parity-check matrices which is not left prime, as shown in Example \ref{ex:degreeleftprime}. 
    In addition, the implication from \ref{itm:cond1} to \ref{itm:cond3} is only true if we assume at least that $H_0$ has full rank (which is a consequence of $H(z)$ being left prime). To see this, consider a parity-check matrix that fulfills \ref{itm:cond3}, i.e. is a parity-check matrix of an MDP convolutional code, and multiply it by $z\mathrm{I}_{n-k}$. The resulting matrix is still a parity-check of the same MDP convolutional code but it has $H_0=0$ and hence, can not fulfill \ref{itm:cond3}.

\end{enumerate}
\end{remark}


\section{Left primeness of parity-check and generator matrices of MDP convolutional codes}\label{sec:mainresult}

In this section, we show for which parameters condition \ref{itm:cond3} of Theorem \ref{thm:characterizationMDP} applied on an $(n,k)$ convolutional code $\C$ for $j=L$ implies that the corresponding parity-check matrix of $\C$ is left prime and thus the degree of $\C$ is equal to the sum of the row degrees of this parity-check matrix.
Moreover, we show for which parameters condition \ref{itm:cond2} of Theorem \ref{thm:characterizationMDP} for $j=L$ implies that the considered convolutional code is noncatastrophic, i.e. for these parameters every MDP convolutional code is noncatastrophic.

\begin{theorem}\label{r1}
Consider $H(z)\in\mathbb F[z]^{(n-k)\times n}$ with $\deg(H(z))=\nu$ and set $\delta=(n-k)\nu$, $r=\left\lfloor\frac{\delta}{k}\right\rfloor$.

If the matrix $$\bar{H}:=\left[\begin{array}{ccc}H_0\\ \vdots & \ddots\\ H_{\nu} & & H_0\\ & \ddots & \vdots\\ & & H_{\nu}\end{array}\right]\in\mathbb F^{(n-k)(r+\nu+1)\times n(r+1)}$$
has full (row) rank, then $H(z)$ is left prime.
\end{theorem}

\begin{proof}

First note that since $(n-k)(r+\nu+1)=n(r+1)+\delta-k(r+1)<n(r+1)$, $\bar{H}$ has more columns than rows. As $\bar{H}$ has full row rank, the map $\mathbb F^{n(r+1)}\rightarrow\mathbb F^{(n-k)(r+\nu+1)}$, $v\mapsto \bar{H}v$ is surjective and there exists $\bar{X}=\begin{pmatrix} X_0\\ \vdots\\ X_r\end{pmatrix}\in\mathbb F^{(r+1)n\times (n-k)}$ with $X_i\in\mathbb F^{n\times(n-k)}$ for $i=1,\hdots,r$ such that $\bar{H}\bar{X}=\begin{pmatrix} I_{n-k}\\0_{n-k}\\ \vdots\\0_{n-k}\end{pmatrix}$. Defining $X(z)=\sum_{i=0}^{r}X_iz^{i}$, one gets $H(z)X(z)=\mathrm{I}_{n-k}$ and hence $H(z)$ is left prime.

\end{proof}

\begin{corollary}
Let $n,k,\delta \in\mathbb N$ with $k<n$ and $(n-k)\mid\delta$ and set $\nu=\frac{\delta}{n-k}$. If $H(z)=\sum_{i=0}^{\nu}H_iz^{i}\in\mathbb F^{(n-k)\times n}[z]$

has the property that all full-size minors of $H_L^c$ with $L=\left\lfloor\frac{\delta}{k}\right\rfloor+\frac{\delta}{n-k}$ that are not trivially zero are nonzero, then $H(z)$ is a left prime parity-check matrix of an $(n,k,\delta)$ MDP convolutional code.
\end{corollary}

\begin{proof}
With the notation of the preceding theorem, one gets $L=r+\nu$ and $\bar{H}$ is a submatrix of $H_L^c$ with the same number of rows. Hence, there is a full-size minor of $\bar{H}$ that is nonzero and $\bar{H}$ has full (row) rank, which additional implies that $H_{\nu}$ is full rank. Consequently, $H(z)$ is left prime and thus, it is the parity-check matrix of an $(n,k,\delta)$ convolutional code, where $\delta$ is equal to the sum of the row degrees of $H(z)$, i.e. $\delta=(n-k)\nu$ as $H_{\nu}$ is full rank. Then, Theorem \ref{thm:characterizationMDP} implies that this code is MDP.
\end{proof}

\begin{remark}\label{g}
With the same reasoning one can show that for $n,k,\delta \in\mathbb N$ with $k<n$ and $k\mid\delta$ and $m=\frac{\delta}{k}$, if $G(z)=\sum_{i=0}^{m}G_iz^{i}\in\mathbb F^{k\times n}[z]$

has the property that all full-size minors of $G_L^c$ with $L=\frac{\delta}{k}+\left\lfloor\frac{\delta}{n-k}\right\rfloor$ that are not trivially zero are nonzero, then $G(z)$ is the generator matrix of a noncatastrophic $(n,k,\delta)$ MDP convolutional code.
\end{remark}

\begin{remark}
The conditions of the preceding theorem, corollary and remark are not necessary (only sufficient) to ensure that the corresponding polynomial matrix is left prime. As mentioned before, a polynomial matrix is left prime if and only if it has a polynomial right inverse and we provided sufficient conditions in order that this is true.
\end{remark}

The following example shows that if $(n-k)\nmid\delta$ (resp. $k\nmid\delta$), then the MDP property on the minors of the sliding parity-check (resp. generator) matrix does in general not imply that the parity-check (resp. generator) matrix of a convolutional code is left prime.

\begin{example}
Let $1\leq \delta<k$ and $\delta<n-k$, i.e. $L=0$. We get that $\deg(H(z))=\lfloor\frac{\delta}{n-k}\rfloor+1=1$, so $H(z) = H_0+H_1z$  and $H_L^c=H_0$. If we choose $H_0$ such that all full-size minors are nonzero and $H_1=-H_0$, then $H_L^c$ fulfills the MDP property \ref{itm:cond3} but $H(z)=(z-1)\mathrm{I}_{n-k}H_1$, i.e. $H(z)$ is not left prime and the degree of the code with this parity-check matrix is zero. Hence, this can not be an $(n,k,\delta)$ MDP convolutional code. Equivalently, we can show that for such code parameters a generator matrix $G(z)=G_0+G_1z$ with $G_0=-G_1$ having all full-size minors nonzero is not left prime but $G_L^c$ fulfills the MDP criterion \ref{itm:cond2}, i.e. $G(z)$ is the generator matrix of a catastrophic $(n,k,\delta)$ MDP convolutional code.
\end{example}

However, it is possible to modify Theorem \ref{r1}, imposing stronger assumptions to get similar results for the case $(n-k)\nmid\delta$. The case of convolutional codes where $(n-k)\nmid\delta$ is deeply investigated in \cite{napp2016constructing}.

\begin{theorem}\label{thm:n-knotdivdelta}
Let $(n-k)\nmid\delta$. Consider $H(z)\in\mathbb F[z]^{(n-k)\times n}$ with $\deg(H(z))=\nu=\lfloor\frac{\delta}{n-k}\rfloor+1$. If there exist $r\in\mathbb N_0$ and a set $S\subset\{1,\hdots, n-k\}$ with cardinality $|S|=\delta - (n-k) \left\lfloor \frac{\delta}{n-k} \right\rfloor$ such that

the matrix $$\bar{H}:=\left[\begin{array}{ccc}H_0\\ \vdots & \ddots\\ H_{\nu} & & H_0\\ & \ddots & \vdots\\ & &\tilde{H}_{\nu}\end{array}\right]\in\mathbb F^{(n-k)(r+\nu+1)\times n(r+1)},$$
where $\tilde{H}_{\nu}$ consists of the $\delta - (n-k) \left\lfloor \frac{\delta}{n-k} \right\rfloor$ rows of $H_{\nu}$ whose row index lies in $S$, has full (row) rank, then $H(z)$ is left prime.
\end{theorem}

\begin{proof}
The result can be shown exactly in the same way as Theorem \ref{r1}.
\end{proof}

\begin{corollary}\label{cor:n-knotdividesdelta}
Let $n,k,\delta \in\mathbb N$ be such that $k<n$ and $(n-k)\nmid\delta$ and set $\nu=\lfloor\frac{\delta}{n-k}\rfloor+1$. Assume that $H(z)=\sum_{i=0}^{\nu}H_iz^{i}\in\mathbb F^{(n-k)\times n}[z]$ is row-reduced and it has the property that all the full-size minors of $H_L^c$ with $L=\left\lfloor\frac{\delta}{k}\right\rfloor+\nu-1$ that are not trivially zero are nonzero. Moreover, assume that there exist $r\in\mathbb N_0$ and $S\subset\{1,\hdots, n-k\}$ such that $\bar{H}$ as defined in the preceding theorem is full row rank. Then, $H(z)$ is  a left prime parity-check matrix of an $(n,k,\delta)$ MDP convolutional code.
\end{corollary}

Note that Theorem \ref{thm:n-knotdivdelta} and Theorem \ref{r1} imply that $H(z)$ has so-called \emph{generic row degrees}, i.e. its row degrees are equal to $\left\lfloor \frac{\delta}{n-k} \right\rfloor +1 $ with multiplicity $t:=\delta - (n-k) \left\lfloor \frac{\delta}{n-k} \right\rfloor$ and $\left\lfloor \frac{\delta}{n-k} \right\rfloor $ with multiplicity $n-k-t$. This ensures that for given parameters $n,k,\delta$, we consider parity-check matrices with the minimal possible degree (as polynomial) $\deg(H(z))=\nu$, where $\nu=\frac{\delta}{n-k}$ if $(n-k)\mid\delta$ and $\nu=\left\lfloor \frac{\delta}{n-k} \right\rfloor + 1$ if $(n-k)\nmid \delta$.\\
Since, according to Remark \ref{g}, for $k\mid\delta$ we can consider the generator matrix of the convolutional code, in the following, we want to investigate for which code parameters with $(n-k)\nmid\delta$ and $k\nmid\delta$ the MDP property \ref{itm:cond3} for $H$ still implies that one can find such $r$ as in Theorem \ref{thm:n-knotdivdelta} and Corollary \ref{cor:n-knotdividesdelta}.

\begin{remark}
In order that it is possible for $\bar{H}$ to have full row rank, it is necessary that 
$$(n-k)(r+\nu+1)-\Big(\delta - (n-k) \left\lfloor \frac{\delta}{n-k} \right\rfloor\Big)-n(r+1)\leq 0.$$

With $\left\lfloor \frac{\delta}{n-k} \right\rfloor=\nu-1$, this is equivalent to
\begin{equation}\label{eq:lowerboundr}
r\geq\frac{2(n-k)\nu-n-\delta}{k}.
\end{equation}

In order that $\bar{H}$ is contained in $H_L^c$, it is necessary that
\begin{equation}\label{eq:upperboundr}
r\leq\left\lfloor \frac{\delta}{k} \right\rfloor-1,
\end{equation}
as for $(n-k)\nmid\delta$ one has $L=\nu-1+\left\lfloor \frac{\delta}{k} \right\rfloor$. 

Here, we also see why it is of advantage to have generic row degrees to keep the degree (as a polynomial) of $H(z)$ small, because only if $r+\nu\leq L$, criterion \ref{itm:cond3} of Theorem \ref{thm:characterizationMDP} implies that $H(z)$ is left prime. 

Combining the upper bound \eqref{eq:upperboundr} and the lower bound \eqref{eq:lowerboundr}, we obtain the condition $$\frac{2(n-k)\nu-n-\delta}{k}\leq \left\lfloor \frac{\delta}{k} \right\rfloor-1. $$ We write $\left\lfloor \frac{\delta}{k} \right\rfloor=\frac{\delta}{k}-\varepsilon_1$ and $\left\lfloor \frac{\delta}{n-k} \right\rfloor=\frac{\delta}{n-k}-\varepsilon_2$, i.e. $\nu=\frac{\delta}{n-k}+1-\varepsilon_2$, with $0<\varepsilon_1, \varepsilon_2<1$. Then, we get that it is necessary to have
$$\frac{\delta+2(n-k)(1-\varepsilon_2)-n}{k}\leq \frac{\delta}{k}-\varepsilon_1-1 $$ and end up with the condition $$\varepsilon_1\leq\left(\frac{n}{k}-1\right)(2\varepsilon_2-1).$$ In particular, this is only possible if $\varepsilon_2>\frac{1}{2}$.

\end{remark}

Interchanging the roles of $n-k$ and $k$, $\varepsilon_1$ and $\varepsilon_2$, as well as of $\nu=\deg(H(z))$ and $m=\deg(G(z))$, we get equivalent results if we consider the generator matrix in the case that $k\nmid \delta$.

\begin{corollary}\label{cor:knotdividesdelta}
Let $n,k,\delta \in\mathbb N$ be such that $k<n$ and $k\nmid\delta$ and set $m=\lfloor\frac{\delta}{k}\rfloor+1$. Assume that $G(z)=\sum_{i=0}^{m}G_iz^{i}\in\mathbb F[z]^{(n-k)\times n}$ is row-reduced and has the property that all the full-size minors of $G_L^c$ with $L=\left\lfloor\frac{\delta}{n-k}\right\rfloor+m-1$ that are not trivially zero are nonzero. Moreover, assume that there exists $r\in\mathbb N_0$ and $S\subset\{1,\hdots, k\}$ with cardinality $|S|=\delta - k \left\lfloor \frac{\delta}{k} \right\rfloor$ such that
$$\bar{G}:=\left[\begin{array}{ccc}G_0\\ \vdots & \ddots\\ G_{m} & & G_0\\ & \ddots & \vdots\\ & &\tilde{G}_{m}\end{array}\right]\in\mathbb F^{k(r+\nu+1)\times n(r+1)},$$
where $\tilde{G}_{m}$ consists of the $\delta - k\left\lfloor \frac{\delta}{k} \right\rfloor$ rows of $G_{m}$ whose row index lies in $S$, has full row rank.
Then, $G(z)$ is the generator matrix of a noncatastrophic $(n,k,\delta)$ MDP convolutional code.
\end{corollary}

Moreover, the MDP property \ref{itm:cond2} for the generator matrix implies the existence of such an $r\in\mathbb N_0$ if
$$\varepsilon_2\leq\left(\frac{n}{n-k}-1\right)(2\varepsilon_1-1).$$ In particular, this is only possible if $\varepsilon_1>\frac{1}{2}$.


\section{Conclusion}\label{sec:conclusion}
In this paper, we proved that
the criterion provided in \cite{gl03} to show that a convolutional code is MDP
implies that the corresponding parity-check matrix is left prime if
$n-k$ divides $\delta$.

The same is true for the generator matrix provided that $k$ divides $\delta$, which implies that all MDP convolutional codes where $k$ divides $\delta$ are noncatastrophic.
 Moreover, when $n-k$ and $k$ do not divide $\delta$, we can add some technical assumption to still get the same result.
 Furthermore, we were able to remove the assumption of left primeness in the known characterization for MDP convolutional codes and show how it is possible to construct an MDP convolutional code with given degree via its parity-check matrix.
 
\section{Acknowledgment}
This work was partially supported by the Swiss National Science Foundation grant n. 188430 and by the German Research Foundation grant LI 3101/1-1. The authors acknowledge also the anonymous reviewer whose comments and suggestions helped improve this manuscript.

\bibliographystyle{abbrv}
\bibliography{references}
\end{document}